\documentclass[journal]{IEEEtran}

\ifCLASSINFOpdf
\else
   \usepackage[dvips]{graphicx}
\fi
\usepackage{url}

\usepackage{graphicx}

\usepackage{cite}
\usepackage{amsmath,amssymb,amsfonts,braket}
\usepackage{amsmath,amsfonts,amsthm,graphicx,gensymb,float,algorithm,braket,xcolor,amssymb,xparse,array,algpseudocode,tikz,nicematrix}
\usepackage{textcomp}
\def\BibTeX{{\rm B\kern-.05em{\sc i\kern-.025em b}\kern-.08em
    T\kern-.1667em\lower.7ex\hbox{E}\kern-.125emX}}
\usepackage{subcaption}

\newtheorem{lemma}{Lemma}
\newtheorem{prop}{Proposition}

\begin{document}

\title{Hamiltonian Encoding for Quantum Approximate Time Evolution of Kinetic Energy Operator}

\author{{Mostafizur Rahaman Laskar}$^1$, {Kalyan Dasgputa}$^2$, {Amit Kumar Dutta}$^1$, {Atanu Bhattacharya}$^3$\\
$^1${G. S. Sanyal School of Telecommunications, Indian Institute of Technology Kharagpur, India}\\
$^2${IBM Quantum, IBM Research, Bangalore, India}\\
$^3${Department of Chemistry, GITAM Visakhapatnam, India}}

\maketitle

\begin{abstract}
The time evolution operator plays a crucial role in the precise computation of chemical experiments on quantum computers and holds immense promise for advancing the fields of physical and computer sciences, with applications spanning quantum simulation and machine learning. However, the construction of large-scale quantum computers poses significant challenges, prompting the need for innovative and resource-efficient strategies. Traditional methods like phase estimation or variational algorithms come with certain limitations such as the use of classical optimization or complex quantum circuitry. One successful method is the Trotterization technique used for quantum simulation, specifically in atomic structure problems with a gate complexity of approximately $\mathcal{O}(n^2)$ for an $n$-qubit realization. In this work, we have proposed a new encoding method, namely quantum approximate time evolution (QATE) for the quantum implementation of the kinetic energy operator as a diagonal unitary operator considering the first quantization level. The theoretical foundations of our approach are discussed, and experimental results are obtained on an IBM quantum machine. Our proposed method offers gate complexity in sub-quadratic polynomial with qubit size $n$ which is an improvement over previous work. Further, the fidelity improvement for the time evolution of the Gaussian wave packet has also been demonstrated. 
\end{abstract}


\IEEEpeerreviewmaketitle

\section{Introduction}

Quantum mechanics, a fundamental theory in physics, provides a powerful framework for understanding the behaviour of particles at the atomic and subatomic levels. In the realm of chemistry, where the properties and interactions of atoms and molecules are of paramount importance, quantum mechanics plays a crucial role in elucidating their behaviour. One key concept in quantum mechanics is the time evolution of quantum states, governed by the unitary operator known as the time evolution operator, denoted as $\mathbf{U}(t)$. This operator describes how a quantum state changes over time, encapsulating the dynamics of a quantum system and allowing for the calculation of various physical observables. In the context of chemistry, the time evolution operator is particularly significant as it enables the simulation and prediction of chemical reactions, the study of energy transfer processes, and the understanding of electronic and vibrational spectra. By leveraging the time evolution operator, chemists can explore intricate details of chemical reactions, including bond breaking and formation, energy transfer, and excited state formation, which are inherently quantum mechanical phenomena. Furthermore, the time evolution operator plays a pivotal role in the field of quantum computing, where it facilitates the simulation and exploration of complex chemical systems, offering promising solutions to computationally demanding problems such as simulating large molecules and optimizing chemical reactions.
\subsection{Background}

Quantum simulation of electronic structure in quantum chemistry has been a prominent research area, aiming to understand the time evolution of wave functions using the kinetic and potential energy operators in the Hamiltonian \cite{klymko2022real, motta2021low, poulin2014trotter, shokri2021implementation}. However, the dynamics of chemical reactions, especially in complex systems, pose challenges that cannot be efficiently addressed by classical computation, necessitating the use of quantum algorithms \cite{dhar2015quantum, stepanov2018quantum, mcardle2019variational}. The complexity is further amplified by interactions between particles and quantum tunnelling effects, which perturb the Hamiltonian operator \cite{tranter2019ordering, stepanov2018quantum, godoy1992quantum, jiang2018quantum}. Despite the demand, the realization of time-evolving states remains challenging due to the high gate complexity of existing quantum simulation algorithms on limited physical resources \cite{shokri2021implementation,malik2019first, pastori2022characterization, kivlichan2018quantum, babbush2017low}.

Quantum Hamiltonian simulations (QHS) involve approximating a unitary operator corresponding to a given Hamiltonian matrix \cite{berry2015hamiltonian, low2019hamiltonian}. Various frameworks, such as the Trotter-Suzuki product formula \cite{berry2007efficient}, truncated Taylor series \cite{berry2015simulating}, qubitization method \cite{low2019hamiltonian}, quantum walk \cite{berry2015hamiltonian}, and optimal quantum signal processing algorithm \cite{low2017optimal}, have been proposed in QHS, each with its own advantages and challenges. 
There are several methods to solve the energy structure problem and their time dynamics on a quantum computer, such as using quantum phase estimation\cite{aspuru2005simulated}, adiabatic algorithm\cite{babbush2014adiabatic}, variational approach\cite{peruzzo2014variational} etc. However, a central challenge remains with the quantum simulation of the underlying Hamiltonian.   
From a practical implementation standpoint, the Trotter-Suzuki-based approach, known as the "Trotterization" technique, is commonly employed on quantum computers for applications like atomic structure problems \cite{poulin2014trotter, babbush2015chemical, shokri2021implementation}. 
The $2^{nd}$-order Trotter-Suzuki approximation for the time evolution of the Hamiltonian $\mathbf{H}$ is given by
\begin{align}
e^{-i\mathbf{H}\Delta t} &= e^{-i(\mathbf{K}+\mathbf{V})\Delta t}\nonumber\\
&= e^{-i\frac{\mathbf{V}}{2}\Delta t} e^{-i\mathbf{K}\Delta t} e^{-i\frac{\mathbf{V}}{2}\Delta t} + \mathcal{O}(\Delta t^3),
\end{align}
where the terms $\mathbf{K}$, and $\mathbf{V}$ represent the Hamiltonian for kinetic energy and potential energy operator respectively. 
In recent literature \cite{mcardle2019variational, motta2020determining}, quantum simulation of the imaginary time evolution for the Hamiltonian operator $\mathbf{H}$ has been considered as a powerful tool for studying quantum systems. In order to implement the time-evolution operator, authors in \cite{mcardle2019variational} showed a variational approach to find the ground state energy of a multi-particle system (e.g., lithium hydride). However, this approach is a hybrid approach which considers a classical optimizer in addition to the quantum circuit to prepare a variational ansatz. An interesting approach known as the inexact quantum imaginary time evolution (QITE) algorithm as shown in \cite{motta2020determining} showed that a unitary operator can be created to a domain $\mathcal{D}$ smaller than that induced by correlations for the resource-limited quantum computation. For studying the dynamics of free particles in a finite potential well, the Trotterization approach is shown promising, especially the implementation on a quantum machine \cite{shokri2021implementation}. However, for the complex Hamiltonian dynamics, circuit optimization has not been explored well, which can pose significant challenges for higher-dimensional configurations\cite{motta2021low}.

\subsection{Contributions:}  Given the above background, our approach is conceptually novel for designing a time evolution operator with the following contributions. 

\begin{itemize}

\item We exploit the bi-symmetric diagonal structure of the kinetic energy operator and propose a quantum pyramid architecture using the ladder of CNOT gates. It uses half of the samples to be encoded using the circuitry and the other half (about the plane of symmetry) is generated by the reflection operator. We have observed that the ladder of CNOT gates acts as a reflection operator. Using this phenomenon, we design a quantum pyramid architecture for encoding the kinetic energy (which is a bi-symmetric diagonal) operator. 

\item We propose new encoding techniques for the kinetic energy operator. One method is quantum approximate time evolution (QATE) encoding for simulating the Hamiltonian with a high accuracy. Here, the $1$ qubit gates requirement is less than the state of the art, however, $2$ qubit gates are similar in number. The other method is called quantum windowing encoding (QWE) which is inspired by the window technique used in signal processing literature. 

\item We implement our proposed algorithms on an IBM quantum machine, and show the results for a Gaussian wave packet evolving in time steps. Further, we show the fidelity result and gate counts for various qubit sizes. We demonstrate new concepts for low-complex simulation of kinetic energy operators, which can have novel applications in the near future.

\end{itemize}

\section{The Time Evolution Operator}

Given a Hamiltonian $\mathbf{H}=\sum_{i=1}^{l} \mathbf{H}_i$, with $l$-local terms, the time evolution of a wave function ${\psi}$ for time-step $\Delta t$ can be written following Schrodinger's equation as 
\begin{align}
    \ket{{\psi}_{t+\Delta t}}= \mathbf{U}(\Delta t) \ket{{\psi}_{t}},
    \label{time_ev}
\end{align}
where $\mathbf{U}(\Delta t)=e^{-i\mathbf{H}\Delta t}$ denotes the time-evolution operator that transforms the state of the system from $|\psi(t)\rangle$ to $|\psi(t + \Delta t)\rangle$. Various mathematical techniques and numerical methods have been developed to compute the Time Evolution Operator efficiently and accurately as discussed in the background. These methods are crucial for simulating quantum systems, studying quantum dynamics, and exploring the behaviour of complex quantum phenomena. Here, we describe a fermionic Hamiltonian system and discuss some special cases with efficient algorithms for implementation on a quantum machine. Our approach is focused on the structural aspects of the underlying Hamiltonian operator to find optimal gate complexity, thereby reducing the total gate cost as well as the noise level in realistic experimentation on a quantum computer. 


In a non-relativistic case, the behaviour of Hamiltonian considers that particles (such as electrons) interact in the external potential of another particle (positively-charged nuclei) described within the Born-Oppenheimer approximation, given by

\begin{align}
    \mathbf{H}= -\sum_{i} \frac{\nabla^2_i}{2}- \sum_{i,j}\frac{q_j}{\vert R_j -r_i \vert} + \sum_{i<j}\frac{1}{r_i-r_j} + \sum_{i<j}\frac{q_iq_j}{R_i-R_j}.
    \label{diff-Ham}
\end{align}
Here, $\mathbf{H}_1=-\displaystyle\sum_{i} \frac{\nabla^2_i}{2}$ denotes the kinetic energy term, $\mathbf{H}_2=- \displaystyle\sum_{i,j}\frac{q_j}{\vert R_j -r_i \vert}$ represents the potential energy where $q_j$ are charges of the nuclei, $R_j$ and $r_i$ are positions of the nuclei and electrons respectively; the $\mathbf{H}_3=\displaystyle\sum_{i<j}\frac{1}{r_i-r_j}$ denotes the electron-electron repulsion potential term, and $\mathbf{H}_4=\displaystyle\sum_{i<j}\frac{q_iq_j}{R_i-R_j}$ is some constant term. Discretization techniques are employed to convert the differential form (\ref{diff-Ham}) to a practical computational problem \cite{babbush2017low,kivlichan2017bounding}, and the Hamiltonian is simplified as
\begin{align}
    \mathbf{H} = \mathbf{K}(\hat{p}) + \mathbf{V}(\hat{x}),
\end{align}
where $\mathbf{K}(\hat{p})$ denotes the discretized kinetic energy operator in momentum domain ($\hat{p}$), and the $\mathbf{V}(\hat{x})$ represents the potential energy term expressed in coordinate domain ($\hat{x}$). The effects due to other terms in (\ref{diff-Ham}) i.e., $\mathbf{H}_3$, and $\mathbf{H}_4$ can either be ignored for simplification or can be absorbed in the corresponding kinetic or potential energy term as per their representation either in position or momentum basis. Note that, we will consider the first quantization level expression for encoding the energy operators in the quantum circuit.

\subsection{Potential Energy}

The potential energy operator denoted as $e^{-i\mathbf{V}\Delta t}$, holds immense importance in quantum mechanics as a key component of the Hamiltonian\cite{griffiths2005introduction, gasiorowicz2005quantum}. It characterizes the potential energy associated with a quantum system, exerting a profound influence on its behaviour and properties. In specific scenarios, the potential energy operator manifests in diverse forms contingent upon the characteristics of the potential energy itself. A noteworthy example is the finite step potential, which features abrupt shifts in potential energy at distinct positions within the system. This phenomenon is relevant in numerous contexts, ranging from quantum wells to barrier structures, where the potential energy undergoes sudden changes at specific locations.

The step potential operator for evolution time $\Delta t/r$ (with $r=2$ for $2^{nd}$-order Trotterization) can be written as
\begin{align}
    \mathbf{U}_{V}({\Delta t/r})=e^{-i\mathbf{V}(\hat{x}) \Delta t/r}.
\end{align}

For a single-step potential, the potential energy abruptly changes at a particular position, given as 
\begin{align}
\mathbf{V}(x) =
\begin{cases}
V_1 & \text{for } x < x_0 \nonumber\\
V_2 & \text{for } x \geq x_0,
\end{cases}
\end{align}

where $V_1$ and $V_2$ represent the potential energy values on either side of the step at $x_0$. The potential energy operator for the single-step potential can be realized by applying the appropriate phase shift based on the potential energy values. Similarly, for a double-step potential, there are two abrupt changes in the potential energy at different positions. Mathematically, this can be represented as
\begin{align}
\mathbf{V}(x) =
\begin{cases}
V_1 & \text{for } x < x_1 \nonumber\\
V_2 & \text{for } x_1 \leq x < x_2 \nonumber\\
V_3 & \text{for } x \geq x_2,
\end{cases}
\end{align}
where $V_1$, $V_2$, and $V_3$ are the potential energy values in different regions separated by the step positions $x_1$ and $x_2$. The potential energy operator for the double-step potential involves applying the respective phase shifts corresponding to each region.
In the case of multiple-step potentials, the potential energy exhibits multiple abrupt changes at different positions. The mathematical description becomes more complex, involving multiple regions with different potential energy values and corresponding phase shifts. To implement the potential energy on a digital computer, we need to perform discretization on $x$-space ($-d<x<d$), where each sample can be represented on a grid (with each smallest grid of $\Delta x=\frac{2d}{N}$ for $N$ samples) as
\begin{align}
    x_k= -d + \left( k+ \frac{1}{2}\right) \Delta x.
\end{align}
The single, double, and multiple well potentials (with equal potential barriers) can be implemented as a quantum circuit using elementary quantum gates as follows
\begin{align}
    \mathbf{U}_{V}^s & =e^{-i \eta Z \Delta t/r} \otimes I \otimes I \dots \otimes I, \nonumber \\
    \mathbf{U}_{V}^d &= I \otimes e^{-i\eta Z\Delta t/r}\otimes I\dots \otimes I ~\text{and}, \nonumber\\
    \mathbf{U}_{V}^m &= I \otimes I \dots I\otimes e^{-i\eta Z\Delta t/r}\otimes I \dots \otimes I,
\end{align}
where $\eta$ denotes magnitude of the potential barrier, and $Z$ is the Pauli-$Z$ operator. 
Note that, by changing the position of the Pauli-$Z$ operator with respect to the identity operators (a single qubit operation) we can create any choice of the step potentials. For a $N\times N$, potential energy operator, here we need a Pauli-$Z$ operator.   Hence, for an $n$ input qubit circuit, the overall number of elementary quantum gates required for the quantum gate implementation of the potential energy operator is given by $\tilde{\mathcal{O}}(n)$ with $N=2^n$.

\subsection{Kinetic Energy}

The unitary form of the kinetic energy operator is expressed as $\mathbf{U}=\mathbf{e}^{-i\mathbf{K}t}$, which is a diagonal matrix representing the time evolution operator. 
Mathematically, we can represent the kinetic energy operator as $\mathbf{K} = \frac{{\mathbf{p}^2}}{{2m}}$, where $\mathbf{p}$ is the momentum operator and $m$ is the mass.

For the implementation of the Kinetic energy operator on a digital computer, we need the discretized representation of the system, where the position is represented by discrete points or samples. In this case, we can represent the momentum variable as $p = n\Delta p$, where $n$ is an integer representing the sample index and $\Delta p$ is the spacing between samples.
Assuming a simple configuration (e.g., the motion of a free particle in a step potential), the kinetic energy operator shows a parabolic function of the momentum variable and exhibits a plane of reflection in about half of the samples (refer to Fig. \ref{res1}.$a$). We consider a one-dimensional grid in the $p$-space, taking a finite range $-d<p<d$ with $N$ uniformly spaced grid-points representing the samples of the continuous variable $p$ given by 
\begin{align}
    p_j:= \frac{\pi}{d}\left( j+ \frac{1}{2}-\frac{N}{2} \right)~~~ \text{for}~ j=0,~1,~\dots, N-1.
    \label{p-discrete}
\end{align}
As the kinetic energy operator is a parabolic function of $p$, it is an even function with respect to the sample index $p_j$, and its diagonal elements exhibit symmetry about the plane of reflection. It shows that the Hamiltonian operator corresponding to kinetic energy has a bi-symmetric pattern due to its parabolic (or even-symmetric) nature. To encode the kinetic energy function in the Hamiltonian operator, we define another variable $\mathbf{D}_{\theta}:=\mathbf{K}(\hat{p})\Delta t$ where $\mathbf{D}_{\theta}=\textbf{diag}(\theta_0,~\theta_1,\dots,\theta_l,\theta_l,\dots,\theta_{N-1})$ with its plane of reflection about the coordinate $(l,l)$ (here, $l=\frac{N}{2}$). 
Using this definition, we can construct the time evolution operator $\mathbf{U}$ as a diagonal matrix with elements $\mathbf{e}^{-i\mathbf{K}t}$ for each diagonal element as follows
    \begin{align}
        \mathbf{U}_K(\Delta t) &= e^{-i\mathbf{D}_{\theta}}\nonumber\\
        &= e^{-i\textbf{diag}(\theta_0,~\theta_1,\dots,\theta_l,\theta_l,\dots,\theta_{N-1})}\nonumber\\
        &= \left(\begin{array}{c@{}c@{}c}
      \mathbf{B} & | & \mathbf{0}\\
      \hline
      \mathbf{0} & | & \mathbf{B}_{ref}
    \end{array}\right),
    \label{ke_eq1}
    \end{align}
where $\mathbf{B}=\textbf{diag}(\mathbf{b})$, and $\mathbf{B}_{ref}$ is the reflection of $\mathbf{B}$ for $\mathbf{b}=[e^{-i\theta_0}, ~e^{-i\theta_1}, \dots, ~e^{-i\theta_l}]$. 
Thus, the unitary relation $\mathbf{e}^{-i\mathbf{K}t}$ through Hamiltonian simulation captures the time evolution of the system under the influence of the kinetic energy operator, with a parabolic momentum dependence and a plane of reflection symmetry about half of the samples.
%
The direct implementation of the operator $\mathbf{U}_K(\Delta t)$ on a superconducting qubit-based quantum machine using the Trotterization method is addressed in \cite{shokri2021implementation, malik2019first}.




\section{Encoding Kinetic Energy Evolution Operator on a Quantum Circuit}

In this research work, we exploit the parabolic nature of the Kinetic energy which successively generates the bi-symmetric pattern in the unitary operator $\mathbf{U}_K$. Here, we define the parameter $\theta_j=\frac{-p_j^2 \Delta t}{2m}$ (for simplicity, we take $m=1$) for $j=0,\dots, N-1$. As a consequence, the $(j,j)^{th}$ coordinate of the $\mathbf{U}_K$ denotes the element $e^{-i\theta_j}$. A conceptually novel quantum architecture is shown here which can efficiently simulate (in an approximate sense) the operator $\mathbf{U}_K$ on quantum hardware by exploiting its structure. The below lemmas demonstrate the motivation behind our approach for the bi-symmetric kinetic energy operator. 

\begin{lemma}\label{lemma1}
    Given $\mathbf{C}$ be a CNOT operator, and $\mathbf{P}=\mathbf{I}\otimes \mathbf{P}_1$ is another operator with $\mathbf{I}$ be the identity operator and $\mathbf{P}_1$ is some phase gate, then $\mathbf{R}=\mathbf{C} \mathbf{P} \mathbf{C}^{\dagger}$ is a bi-symmetric diagonal quantum operator. 
\end{lemma}

\begin{proof}
  The proof is given in Appendix-\ref{bisym-lem1}.
\end{proof}

\begin{lemma}\label{lemma2}
 Given a list of phase gates as $\mathbf{P}_1,~\dots,~\mathbf{P}_n$ with every $\mathbf{P}_j=diag([1 ~ e^{i\theta_j}])$ placed at $j^{th}$ qubit starting $q[1]$ (second qubit) to $q[n-1]$ (last qubit) with $n=\log_2 N$, the product of the operators $\mathbf{F}_1,~\dots, ~\mathbf{F}_n$ is a diagonal matrix of dimension $N\times N$ with first $2^{n-1}$ elements repeated in order along the main diagonal, where every $\mathbf{F}_j$ is obtained by placing $\mathbf{P}_j$ phase gate at $j^{th}$ qubit in absence of any other gates. 
\end{lemma}
\begin{proof}
  The proof is given in Appendix-\ref{bisym-lem2}.
\end{proof}

Often, the kinetic energy of a fermion has the form of a quadratic function (parabolic), which can be represented as a bi-symmetric and diagonal unitary operator. Based on Lemma-\ref{lemma1}, and Lemma-\ref{lemma2}, we have given the below proposition for the bi-symmetric diagonal operator. 

\begin{prop}\label{prop1}
  The operator $\mathbf{A}=\mathbf{I}\otimes \dots \otimes (\ket{\mathbf{0}}\bra{\mathbf{0}} \times \mathbf{I}~ + ~\ket{\mathbf{1}}\bra{\mathbf{1}}\times \mathbf{X})\otimes\mathbf{I}\otimes\dots \otimes\mathbf{I}$ is a row-exchange operator for a given matrix $\mathbf{F}$ (assuming compatible with $\mathbf{A}$) and a Pauli-operator $\mathbf{X}$, if it is multiplied as $\mathbf{AF}$, and $\mathbf{A}^{\dagger}=\mathbf{A}$ is a column-exchange operator when it is post-multiplied as $\mathbf{FA^{\dagger}}$. The product $\mathbf{A}\mathbf{F}\mathbf{A}^{\dagger}$ has a symmetry about the mid-point along the main diagonal when $\mathbf{F}=\mathbf{F}_1,\dots, \mathbf{F}_n$ following lemma-\ref{lemma2}.
\end{prop}

\begin{proof}
    The proof is given in Appendix-\ref{propsym1}.
\end{proof}


\subsection{Proposed Algorithm}

The Trotterization algorithm is employed for the overall time-evolution algorithm design. The pseudo-code for the $2^{nd}$ order Trotter-Suzuki method is given in Algorithm-\ref{trot-algo}.
\begin{algorithm}
\caption{Trotterization-based Time Evolution Operator}
\begin{algorithmic}[1]
\Procedure{Trotterization}{$\mathbf{K}$, $\mathbf{V}$, $\Delta t$}
    \State $\mathbf{U} \gets \text{Identity Matrix}$
    \State $t \gets t_0$ 
    \State $\ket{\psi_{t}}\gets \psi_{t_0}$ \Comment{Initialization}
    \For{$t \gets t+\Delta t$} \Comment{Apply Trotter-Suzuki approximation}
        \State $\mathbf{U}  \gets e^{-i\mathbf{V}t/2} \cdot \mathbf{U}_{QFT} e^{-i\mathbf{K}t}\mathbf{U}_{IQFT} \cdot e^{-i\mathbf{V}t/2} \cdot$ \Comment{Apply Trotterization}
        \State $\ket{\psi_{t+\Delta t}}\gets \mathbf{U} \ket{\psi_{t}}$
    \EndFor
    \State \textbf{return} $\mathbf{U}$, $\ket{\psi_{t+\Delta t}}$
\EndProcedure
\end{algorithmic}
\label{trot-algo}
\end{algorithm}
Note that, in the Trotterization method we employ the quantum Fourier transform (denoted as $\mathbf{U}_{QFT}$) to transform the Kinetic energy from momentum basis ($\hat{p}$) to space basis ($\hat{x}$). The implementation of potential energy operator term $\mathbf{U}_P=e^{-i\mathbf{V}t/2}$ can be implemented in linear gate complexity with input qubit size as discussed earlier. Employing the Trotterization technique to implement the kinetic energy term $\mathbf{U}_K=e^{-i\mathbf{K}t}$ on a quantum machine requires significant quantum resources. Here, we propose a new quantum architecture, namely quantum pyramid architecture (QPA) which helps us design a bi-symmetric operator, which is often the case for kinetic energy operators. A natural operator representation of the kinetic energy as a function of momentum has a plane of reflection about the skew-diagonal as shown in (\ref{ke_eq1}). Exploiting this structure, our proposed quantum architecture is shown in Algorithm-\ref{algo-QPA}. 
\begin{algorithm}
\caption{Proposed Algorithm for QPA}
\begin{algorithmic}[1]
\Procedure{QPA}{$n,\boldsymbol{\theta}$}
    \State $QR \gets QuantumRegister(n)$
    \State $CR \gets ClassicalRegister(n)$
    \State $QC \gets QuantumCircuit(QR,CR)$ \label{QC_init}\Comment{Quantum Circuit Initialization}
    \For {$j \gets 1:n-1$}
        \State $QC.Cx(QR[0],QR[n-j])$ 
        \Comment{Pyramid Layer $1$}
        \State Encode the Operator Circuit for $\mathbf{b}(\boldsymbol{\theta})$
        \State $QC.Cx(QR[0],QR[j]))$ 
        \Comment{Pyramid Layer $2$}
    \EndFor
    \State \textbf{return} $\mathbf{U}_K$
\EndProcedure
\end{algorithmic}
\label{algo-QPA}
\end{algorithm}

\textit{Note on Algorithm-\ref{algo-QPA}:}
Here, the inputs are number of registers ($n$), and the phase vector $\boldsymbol{\theta}=\theta_0~,\dots,~\theta_l$. The quantum pyramid architecture can be designed in a ladder-cascaded form as shown in Fig. \ref{qpa-fig} following the QPA pseudo-code as described. Here, $QR,~CR$, and $QC$ denote the number of quantum registers, classical registers and the quantum circuit respectively. Here, $Cx$ denotes the controlled-NOT gate, which creates entangled quantum states in the circuit. We demonstrate two methods for the encoding of the Hamiltonian. The first encoding method approximates the kinetic energy operator with $n-1$ phase gates, and $^{n-1}\mathcal{C}_2$ controlled-phase gates. The complexity can be further reduced to $\mathcal{O}(n)$ for certain experiments where $n$-level energies are studied instead of $n^2$ available bands, with a proposed approach called quantum windowing encoding (QWE). 

\begin{figure*}[htb!]
    \centering
    \includegraphics[width=0.8\linewidth]{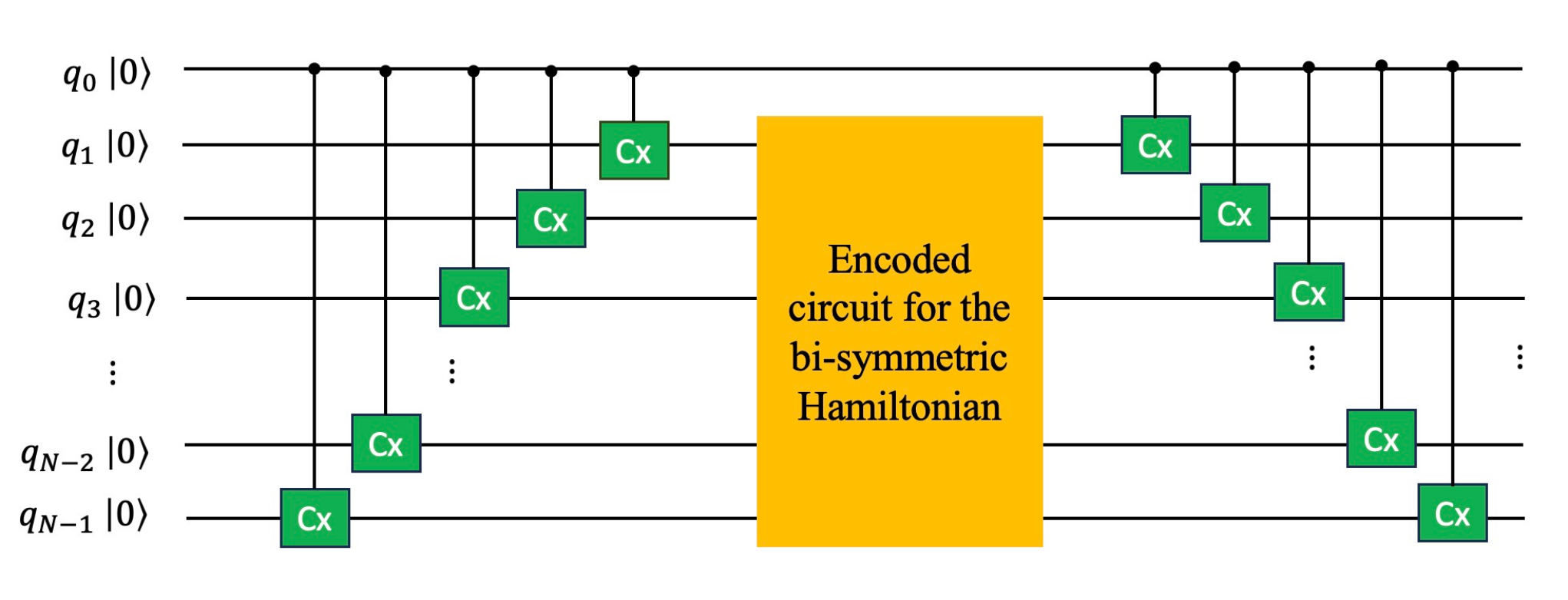}
    \caption{Proposed quantum pyramid architecture for bi-symmetric Hamiltonian operator}
    \label{qpa-fig}
\end{figure*}

\subsection{Encoding Method}

The QPA algorithm helps us to encode the Hamiltonian for $\frac{N}{2}$ sampling points, instead of $N$ samples. The $\frac{N}{2}$ phase samples stored in the vector $\boldsymbol{\theta}$ needs to be encoded in the matrix exponential as $e^{-i\boldsymbol{\theta}}$ along the diagonal of the operator $\mathbf{U}_K$. In the encoding procedure, we will be using $1$-qubit phase gates (denoted as $p$) and $2$-qubit controlled-phase gates (denoted as $Cp$). 

%
Here, we propose two encoding techniques as follows. 

\subsubsection{Quantum Approximate Time Evolution (QATE)}

The quantum approximate time evolution (QATE) encoding method uses $n-1$ number of phase gates, and $^{n-1}\mathcal{C}_2$ controlled phase gates to approximate the matrix exponential for all $2^{n-1}$ samples of phases. The pseudo-code of the QATE algorithm is given below. 

\begin{algorithm}[htb!]
\caption{Proposed QATE Encoding}
\begin{algorithmic}[1]
\Procedure{QATE}{$n$}
    \State $\boldsymbol{\theta}_P \gets$ Primary angles
    \State $\boldsymbol{\theta}_C \gets$ Composite angles
    \For {$i \gets 0:n-1$}
        \State $QC.P(-\boldsymbol{\theta}_P[i],QR[i+1])$ \Comment{Phase gate encoding}
    \EndFor
    \For {$i, x$ in enumerate(sequence)}
        \State $index = i$
        \State $QC.Cp(-\boldsymbol{\theta}_C[index], QR[x[0]], QR[x[1]])$ \Comment{Controlled-phase gate encoding}
    \EndFor
    \State \textbf{return} $\mathbf{U}_K$
\EndProcedure
\end{algorithmic}
\label{QATE-algo}
\end{algorithm}

\textit{Note on Algorithm-\ref{QATE-algo}:} In the QATE encoding method, we first allocate the $\theta_0$ sample as a global phase in the initialization of the algorithm, which does not require any additional resources. The working principle of the QATE algorithm is discussed as follows.

\begin{itemize}
    \item We will consider a $n$ qubit circuit, to encode the kinetic energy function $\mathbf{K}$ in the diagonal unitary matrix form as an evolution operator. Here, the parameters ${\theta_j=-\frac{p_j^2 \Delta t}{2}}$ are the samples of the Kinetic energy samples multiplied by the evolution time. 

    \item Using a QPA architecture, one can see some indices of the unitary matrix can be uniquely prepared by placing phase gates in the quantum circuits. The other indices of the unitary diagonal matrix are composed of the linear combinations of those phase angles. Based on these observations, we divide the set of angles $\{\theta[j]\}$ into two sets. One set is called primary angles, which is stored in an array $\boldsymbol{\theta}_P$. Note that, we will subtract the global phase here, in case the first phase component (i.e., $\theta[0]$) is encoded as a global phase. 

    \item We have observed that if we assign phase gates from the second qubit onward till the last qubit (from up to down approach), the positions (or indices along the diagonal) which are uniquely represented with those phase angles ( while other indices are linear combinations of the phase angles due to the tensor product representation of all gates) can be found as follows:

    \begin{align}
        \begin{cases}
    2^k, & \text{for}~~ k=0: n-1\\
    2^k -1,& ~\text{for}~~{k= n}\label{seq}. 
    \end{cases}
    \end{align}

    For example, if $n=3$, the unique functional values in the diagonal of the overall unitary operator can be found at indices $2,4,7$.

    \item Now, suppose one is interested in encoding the vector in the principal diagonal of the unitary matrix as $1:1$ correspondence with the angles ${\theta_j}$ with the functional value ${e^{-i\theta_j}}$ for $j$ varies from $0:2^n-1$. The first functional value is already encoded due to global phase $\theta_0$, which is performed to eliminate any gate requirement for encoding the first functional point $e^{-i\theta_0}$. As a consequence, all other functional values are biased with the angle, which needs to be subtracted in successive encoding. To perform this manipulation, we define dummy variables $\{a_i\}$ in relation to the phase angles $\{\theta_j\}$.   
    
    \item Our first approach is to encode the primary angles, which are $\{\theta_j\}$ where $j$ follows \eqref{seq}. They can uniquely encode functional values. Accordingly, we define the dummy variables (with their suffix named with unique indices) $\{a_i\}$ by adjusting with the global variables. For the $n=5$ qubit system, one can follow \eqref{encod_fig}. Here, we choose $\theta_2,\theta_4,\theta_8,\theta_{15}$ to encode their corresponding functional values in the diagonal unitary operator. Note, that the QPA architecture takes half of the samples (here $16$ samples are considered instead of $32$) and reflects the other half by exploiting the bi-symmetric structural advantage of the operator.  

    \item Now, already the other indices in the principal diagonal vector of the unitary matrices are impacted by the primary angles. We have observed that for $n=5$, the positions $6,7,10,11,12,13$ are impacted in the diagonal which are in fact related to the combinations of the primary angles. One can easily find this relation for $n$ qubit system by training several values of $n$. Note that, as these positions are impacted, we can manipulate them with entangled controlled phase gates and retain the actual functional values. For example, for $n=5$, one can see in \eqref{encod_fig}, we have encoded $\theta_6,\theta_7,\theta_{10},\theta_{11}, \theta_{12}, \theta_{13}$ with their functional correspondence in the diagonal of the unitary matrix. We have used the dummy variables $a_6,\dots,a_{13}$ here for adjusting the primary angles and the global phase. We call these angles the composite angle. 

    \item All the primary angles are stored in an array $\boldsymbol{\theta}_P$ and all the composite angles needs to be stored in another array, namely $\boldsymbol{\theta}_C$ which are used in the QATE algorithm. Note, that with QPA architecture, the QATE encoding requires $n-1$ phase gates, and there will be $^{n-1}C_2$ possible combinations of controlled phase gates. Hence, the size of $\boldsymbol{\theta}_P$ is $n-1$, and it is $^{n-1}C_2$ for $\boldsymbol{\theta}_C$.

\end{itemize}


One can find (by parity checking of the binary representation of the indices where the qubits are placed) that, there is a relation between the primary angles and composite angles related to the placement of phase gates in certain qubits. We have found that phase angles for the $Cp$ gates can be obtained as a function of the sequence generated from primary angles. For example, in case $n=3$, the phase of the $Cp$ gate will be $\boldsymbol{\theta}_C[1]=\theta[1] - (\boldsymbol{\theta_P}[1]+\boldsymbol{\theta_P}[2])-\theta[0]$, where $\theta[0]$ is the global phase, $\boldsymbol{\theta_P}[1],\boldsymbol{\theta_P}[2]$ are primary angles. QATE algorithm is further explained in the result section.

\subsubsection{Quantum Windowing Evolution (QWE)}

One may find interest in the study of the time evolution operator for certain regions of interest in the momentum domain. It can be either estimating the portion of the wave function where the probability amplitudes are at their peak or at locations where the amplitudes are low.  We may often seek to know in certain regions of the lattice in the momentum domain and the corresponding kinetic energy (KE) time-evolution operator that accurately gives only those portions. In such a scenario, we can reduce the quantum gate complexity to a linear scale, i.e., $\mathcal{O}(n)$, by windowing the region of interest to certain lattice points for a close approximation of the evolution operator. The pseudo-code for the proposed quantum windowing evolution (QWE) encoding is described as follows. It is to be noted that the windowing operation is being done in the momentum domain, where the KE operator operates.

\begin{algorithm}[htb!]
\caption{Proposed QWE Encoding}
\begin{algorithmic}[1]
\Procedure{QWE}{$n$,$\boldsymbol{\theta}_W$, $K$}
    \If{$k \in {K_1} < K$ is in $\boldsymbol{\theta}_W$}
        \State Sort $\boldsymbol{\theta}_{{K_1}}$
        \For{$\textbf{each}$ $k \in {K_1}$}
            \State perform $QC.P(\boldsymbol{\theta}_{{K_1}}[k], QR[k])$
        \EndFor
    \Else
        \State Sort $\boldsymbol{\theta}_m$ for $m = K - {K_1}$
        \For{$k = 1$ to $n$}
            \State perform $QC.Cp(\boldsymbol{\theta}_m[k], QR[1], \boldsymbol{\theta}[m-1])$
        \EndFor
    \EndIf
    \State \textbf{return} $\mathbf{U}_K$
\EndProcedure
\end{algorithmic}
\label{QWE-algo}
\end{algorithm}

\textit{Note on Algorithm-\ref{QWE-algo}:}

In the QWE encoding, we decide a window of samples having length $K \in  \tilde{\mathcal{O}}(n)$, denoted as $\boldsymbol{\theta}_W \subset \boldsymbol{\theta}$. We find the angles $\theta_k$ ($k=1,\dots, K_1<K$) for which $e^{-i\theta_k}$ is an element in the diagonal of $\mathbf{D}_{\theta}$, and we sort all such angles in $\boldsymbol{\theta_{{K_1}}}$ which can be created by the phase gates. Similarly, we find the remaining $m=K-{K_1}$ positions which can be created as entangled positions by the CNOT gates and sort them as $\boldsymbol{\theta}_m$. The windowing encoding method follows similar embedding as in QATE encoding to prepare the quantum circuit. However, the total number of gates required in this procedure is kept within $\mathcal{O}(n)$ to realize $n$-amplitudes instead of the $n^2$ found in the entire lattice. However, the QWE encoding may not be used for the time evolution of a wave packet in the displacement domain. The QWE algorithm can be a low-cost version of the QATE algorithm, where one can play with how many $Cp$ gates need to be placed in the circuit thereby trade-off between complexity and accuracy.

\section{Results and Discussions}

In this result section, we demonstrate numerical simulation results performed on an IBM quantum machine and quantum simulator. In the below subsections, we show the implementation of the kinetic energy operator for the $5$ qubit system as an example using the QATE and QWE encoding method. Further, we portray the time evolution of a Gaussian wave function following (\ref{time_ev}) with the proposed QPA and QATE encoding method (compared with classical simulation). The performance of the proposed framework of the time evolution operator is measured with fidelity and complexity as the key parameter indices (KPIs) and also compared with the state-of-the-art method. In the below Table-\ref{table_param}, we have shown the choice of parameters taken in the simulation environment. 
\begin{table}[htb!]
    \centering
\begin{tabular}{|c|c|}
    \hline
    Quantum Simulator & Statevector Simulator, Qasm Simulator \\
    \hline
    Number of qubits ($n$) & $3-10$\\
    \hline
    Number of shots & $1000-10000$\\
    \hline
    Evolution time ($\Delta t$) & $0.1$ second\\
    \hline
    Range of space coordinate ($x$) & $[-10, ~10]$ in $A^o$ (Angstrom)\\
    \hline
    Sampling interval ($dx$) & $0.625$\\
    \hline
    Wave packet encoding method & amplitude encoding\\
    \hline 
\end{tabular}
    \caption{Simulation parameters}
    \label{table_param}
\end{table}

\subsection{Proposed Experimental Procedure for Kinetic Energy Operator Design}

The algorithm in \ref{algo-QPA} shows how a pyramid-like architecture of CNOT gates helps to design a bi-symmetric operator. In our case, the kinetic energy operator has a plane of reflection about the skew-diagonal. As a consequence, the first $\frac{N}{2}$ elements of the diagonal matrix are the reflection of the second $\frac{N}{2}$ elements of the matrix. However, constructing the matrix with desired functional values in $a_{i, i}$ positions for $i\in [N]$ requires the proper choice of phases in the one-qubit phase gates ($P$) and two-qubit controlled-phase gates ($Cp$) while placing on a particular qubit in the circuit. Here, we demonstrate the experimental set-up for $5$ qubit as an example.

The discrete kinetic energy values in vector form $\mathbf{k}=\left[KE(p_0), ~KE(p_1),~\dots, KE(p_{31})\right]$ can be encoded in angles $\boldsymbol{\theta}=\left[\theta_0,\theta_1,\dots,\theta_{31}\right]$ for time segment $\Delta t$ as $\boldsymbol{\theta}=\mathbf{k}\Delta t$. As a consequence, the unitary operator for the Hamiltonian operator $\mathbf{k}$ denoted as $\mathbf{U}_k=e^{-i \mathbf{diag (k)}\Delta t}$ becomes a function of $\boldsymbol{\theta}$, expressed as 
    $\mathbf{U}_k := \exp{-i \mathbf{diag}\boldsymbol{(\theta)}}$.
As $\mathbf{U}_k$ is a bi-symmetric operator about its skew-diagonal, we design a quantum circuit for the first half of samples (i.e., $\theta_0,\dots, \theta_{15}$) using the QPA algorithm. With QPA, we have implemented the kinetic energy operator using the combination of phase gates and controlled phase gates. For the $5$ qubit quantum circuit, our choice of the phase angles (after adjustment with global phases and employing the QATE encoding method) are given in (\ref{encod_fig}).

\begin{figure*}[htb!]
\hrule
\begin{align}
    a_0 &:= \theta_0 ~~~(\text{can be assigned as global phase})\nonumber\\
    a_2 &:= \theta_2 - a_0 ~~~(\text{ phase for $QR[1]$})\nonumber\\
    a_4 &:= \theta_4 - a_0 ~~~(\text{phase for $QR[2]$})\nonumber\\
    a_8 &:= \theta_7 - a_0 ~~~(\text{phase for $QR[3]$})\nonumber\\
    a_{15} &:= \theta_{15} - a0~~~(\text{phase for $QR[4]$})\nonumber\\
%
%
a_6&:= \theta_6 - (a2+a4) - a0 ~~~(\text{phase for the $Cp$ gate between $QR[1]$, and $QR[2]$})\nonumber
\\
a_7&:= \theta_7 - (a8+a15) - a0 ~~~(\text{phase for the $Cp$ gate between $QR[3]$, and $QR[4]$})\nonumber
\\
a_{10}&:=\theta_{10}-(a2+a8) - a0~~~(\text{phase for the $Cp$ gate between $QR[1]$, and $QR[3]$})\nonumber
\\
a_{11}&:= \theta_{11} - (a4+a15) - a0~~~(\text{phase for the $Cp$ gate between $QR[2]$, and $QR[4]$})\nonumber
\\
a_{12}&:= \theta_{12} - (a4+a8) - a0 ~~~(\text{phase for the $Cp$ gate between $QR[2]$, and $QR[3]$})\nonumber
\\
a_{13}&:=\theta_{13} -(a2+a15) - a0~~~(\text{phase for the $Cp$ gate between $QR[1]$, and $QR[4]$}).
\label{encod_fig}
\end{align}
\hrule 
\end{figure*}

Using the above angles following the QATE procedure, we simulate the Kinetic energy operator for $5$ qubit quantum circuit on an IBM machine using the 'Statevector' quantum simulator as shown in Fig. \ref{qubit5}. Here, we have used $4$-phase gates and $6$ number of $Cp$ gates to simulate the operator $\mathbf{U}_K \in \mathbf{C}^{32\times 32}$. 

\begin{figure*}[htb!]
\centering
    \colorbox[RGB]{239,240,241}
  {\begin{minipage}{.97\linewidth}
    \centering
    \includegraphics[width=\linewidth]{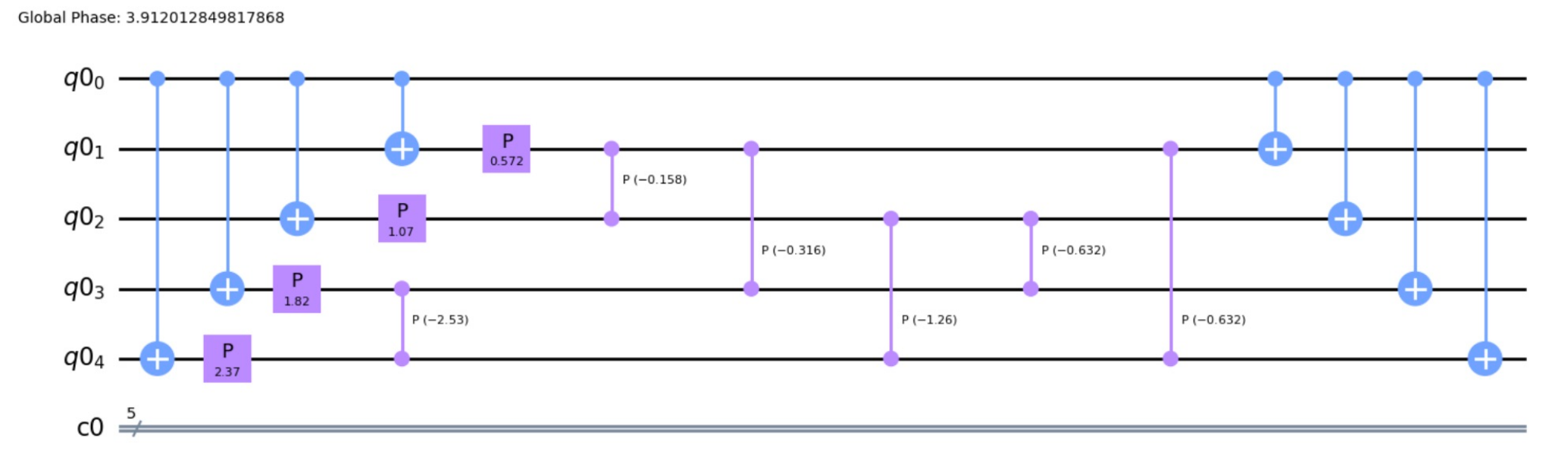}
    \caption{Kinetic energy evolution operator designed for $5$ qubit system with proposed QPA and QATE encoding using QISKIT script on IBM 'Statevector' quantum machine.}
    \label{qubit5}
    \end{minipage}}
\end{figure*}
One can simulate the Hamiltonian $\mathbf{U}_k$ using a classical procedure with $2^n$ samples of angles (note that here angle refers to the quantity "kinetic energy $\times$ time" in the equation $e^{-i \mathbf{K} t}$). The discretized kinetic energy as a function of momentum is shown in Fig. \ref{res1}.$a$. We show the plot of the diagonal array in the quantum simulated unitary matrix $\mathbf{U}_k$ designed for $5$ qubits in Fig.\ref{res1}.$b$. Note that, the QPA algorithm using the QATE encoding technique simulates the kinetic energy operator arbitrarily close to the classical simulated operator, which are overlapped in the given figure. However, the QWE encoding technique can only simulate the part of an operator within our region of interest with a lesser number of quantum gates. 

\begin{figure*}[htb!]
\centering
    \colorbox[RGB]{239,240,241}
  {\begin{minipage}{.97\linewidth}
  \centering
    \includegraphics[width=0.95\linewidth]{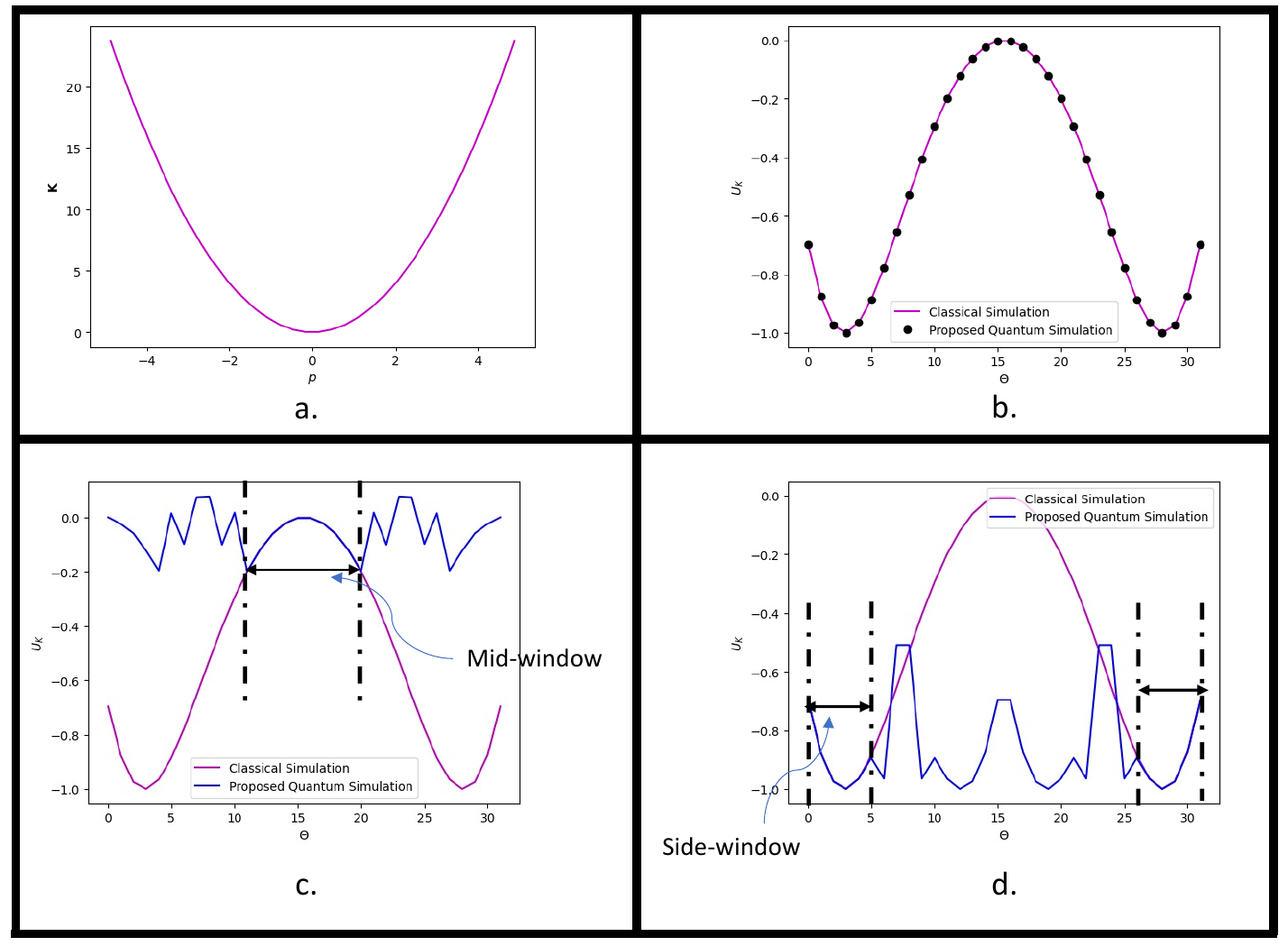}
    \caption{\textbf{Realization of the QATE and QWE algorithm with $5$ qubit register:} \textbf{a.} Plot of kinetic energy ($\mathbf{K}$) as a function of momentum ($p$), \textbf{b.} Simulation of Kinetic energy operator with QATE algorithm. Here, four-phase gates and five $Cp$ gates are used to simulate $\mathbf{U}_{KE}$. Note that, the quantum simulation is arbitrarily close to the classical simulation result using the QATE procedure. \textbf{c.} Simulation with windowing encoding procedure is performed in the mid-region of the Kinetic energy function. Here, we realize four amplitudes near the mid-windowed evolution $\mathbf{U}_{KE}$ using four phase gates and one $Cp$ gate. \textbf{d.} Quantum simulation is performed for the side window of the wave function using three-phase gates and one $Cp$ gate.} 
    \label{res1}
    \end{minipage}}
\end{figure*}

In Fig.\ref{res1}.$c.$, we have shown the mid-windowing method. Here, we estimate $4$ amplitudes in the mid-region of the evolution operator. For the mid-windowing evolution, our choice of the angles (with adjustment following QWE algorithm) are as follows: $a_0=\theta[0],~a_{15}= \theta[15]-a_0,~a_{11}= \theta[11]-a_{15}-a_0,~a_{12}= \theta[12]-a_{11}-a_0,~a_{13}= \theta[13]-a_15-a0,$ and $~a_{14}=\theta[14]-a_{11}-a_{12}-a_{13}-a_0$. Here, phase gates are placed (following QWE encoding) from second to fifth qubit register with phases $-a_{13},-a_{11},-a_{12},-a_{15}$ respectively, and a controlled phase gate is applied between second and third qubit register to create an entanglement with phase $-a_{14}$. 
Similarly, we show side-window encoding when we are interested in the terminal region of the evolution operator (e.g., near the valence energy states) in Fig. \ref{res1}.$d$. 
Here, we encode the side windowing evolution operator with the angles $a_0=\theta[0], ~ a_1=\theta_2-a_0,~a_4=\theta[4]-a_0,$ and $a_1=\theta[1]-a_2-a_3-a_4-a_0$. Here, we place the phase gates from the second qubit to the fourth qubit with angles $-a_2,~-a_4,~-a_3$ respectively, and we place a controlled phase gate with angle $-a_1$ between the second and third qubit registers. Note that, here we have used $\theta[k]$ and $\theta_k$ interchangeably with the same notion of a sample of phase at $k^{th}$ instant.  
\begin{figure*}[htb!]
\colorbox[RGB]{153,255,255}
  {\begin{minipage}{.97\linewidth}
    \centering
    \includegraphics[width=0.98\linewidth]{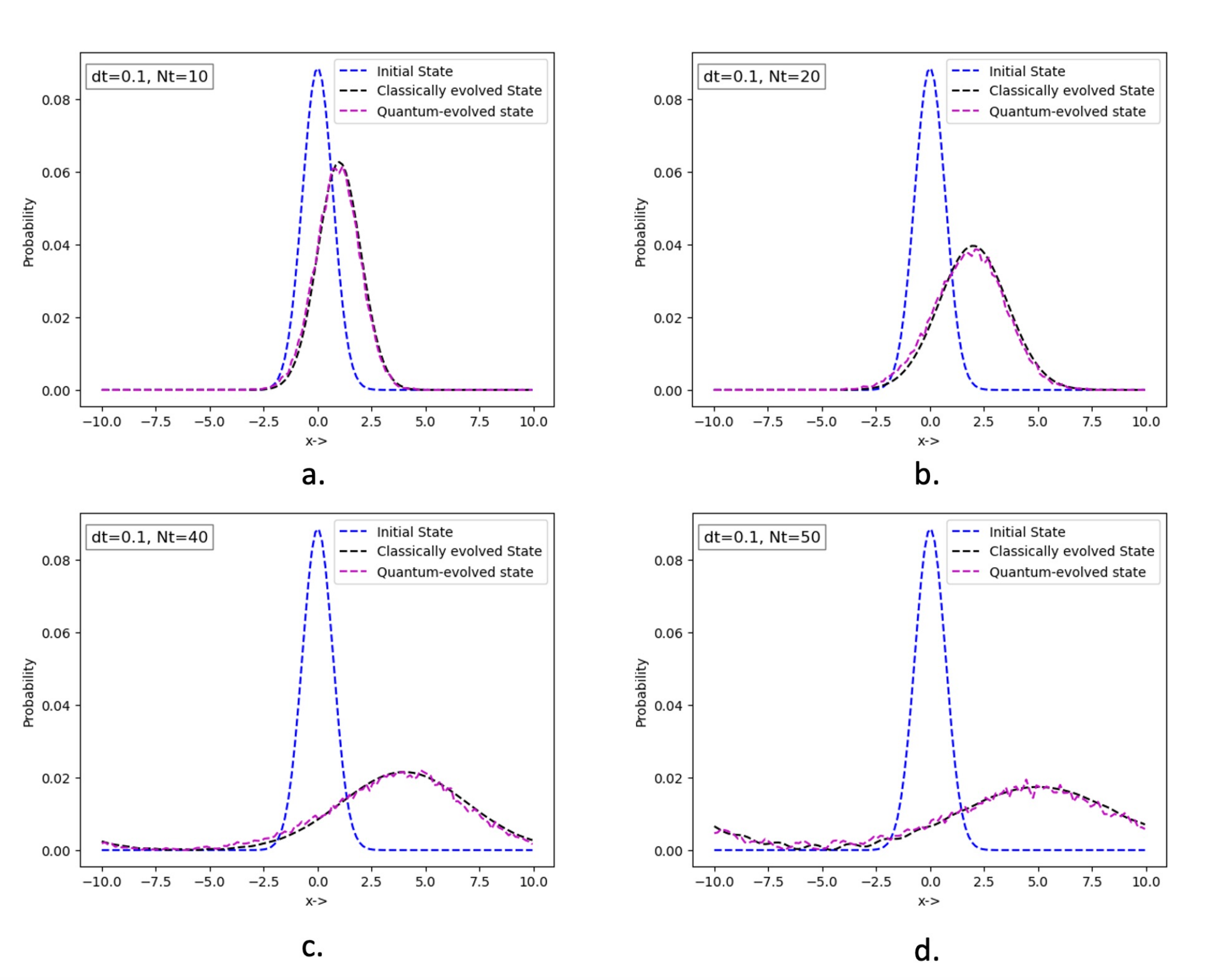}
    \caption{Time evolution of a Gaussian wave packet in the presence of kinetic energy operator in a unit step potential well: Time evolution is shown for an evolution time of $\Delta t$, and varying Trotterization step size $Nt$. The figure denotes in \textbf{a.} time evolution with $Nt=10$ steps, \textbf{b.} time evolution with $Nt=20$ steps, \textbf{c.} time evolution with $Nt=40$ steps, and \textbf{d.} time evolution with $Nt=50$ Trotterization steps. The quantum evolved state obtained with proposed QPA and QATE encoding for $5$ qubits approaches the classically evolved state. Here, the unit of distance is in Angstrom.}
    \label{res2}
    \end{minipage}}
\end{figure*}

A Gaussian wave packet with the form $\boldsymbol{\psi}_0=e^{-\frac{x^2}{2}}e^{ik_0x}$ is considered for the study of its time evolution dynamics which is often chosen as initialization \cite{goldberg1967computer, garraway1995wave}. The wave function is normalized and embedded as the initial quantum state in the qubit registers using the amplitude encoding method. The kinetic energy operator $\mathbf{U}_k (\hat{p})$ designed with the proposed QPA method and QATE encoding is applied on the initial state to get the final state $\boldsymbol{\psi}_t$. As the kinetic energy is a function of the momentum ($p$), and the wave function is defined in terms of the space coordinate ($x$), we employ the quantum Fourier transform (QFT) and its inverse (IQFT) to represent the overall operator in the displacement domain (space coordinate) as 
\begin{align}
    \boldsymbol{\psi}_t= \mathbf{U}_{QFT}\mathbf{U}_k(\hat{p})\mathbf{U}_{IQFT} \boldsymbol{\psi}_0.
\end{align}
In Fig. \ref{res2}, we have shown the time evolution of the Gaussian wave packet performed on IBM 'Qasm simulator' with $5$ qubit registers. We have considered an evolution time of $\Delta t=0.1$, and varying Trotterization steps ($10-50$). The simulation is performed for $10000$ quantum shots to get the probability histogram of the measurement bases ($00000$ to $11111$). With $5$ qubit registers, the quantum-evolved state approaches the classically evolved state with fidelity of $0.73$ approximately. However, by increasing the qubit size the fidelity can be further improved as discussed in the next subsection. Here, we have seen that the quantum-evolved state with our proposed quantum framework is very near to the classically evolved state which has potential usages for the study of dynamics of various wave functions in physics and chemistry. Note that, the unit of time ($\Delta t$) and space ($dx$) for the study of atomic or orbital energy levels may be in atomic unit ($au$).

\textit{\textbf{Note:}} For the time evolution of a wave packet in the coordinate domain, we need to perform the QFT (to transform the Kinetic energy from the momentum domain to the coordinate domain). For this, we rely on the QATE encoding procedure which encodes the kinetic energy for the entire domain of interest. However, the QWE encoding technique is limited within the momentum domain for the applications within small regions of interest. Applying the QFT in the quantum circuit with QWE encoding may not provide perfect time evolution. Also, the application of the time evolution operator in the momentum domain directly is also limited fat present.


\subsection{Fidelity comparison}
\begin{figure}
    \centering
    \includegraphics[width=0.9\linewidth]{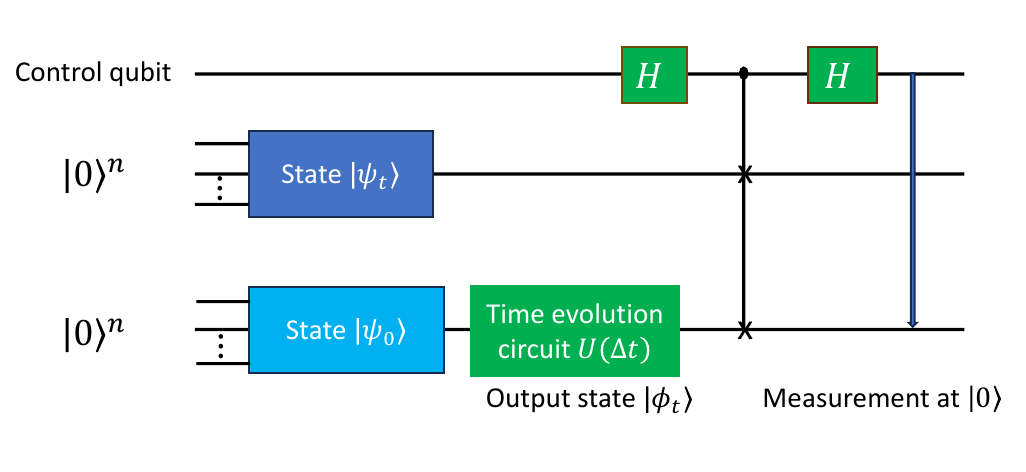}
    \caption{Quantum swap-test circuit}
    \label{swap_ckt}
\end{figure}

We perform a quantum fidelity test using the quantum swap circuit to measure the accuracy in terms of the inner product $\braket{\psi|\phi}$, where $\psi$ is the actual (or target) state and $\phi$ denotes the estimated (or output) state. The swap test circuit takes two input states $\psi_t$, and $\phi_t$ and outputs a probability in computational basis $\ket{0}$ as 
\begin{align}
    Pr(first ~qubit=0) = \frac{1}{2} + \frac{1}{2}\vert \braket{\psi_t|\phi_t} \vert^2. 
\end{align}
In Fig. \ref{swap_ckt}, we have shown a schematic of the swap-test circuit to measure the distance between the quantum evolved state ($\phi_t$) and the target state ($\psi_t$). The target state is a quantum state which is obtained by doing amplitude encoding of the classically evolved state. The output state $\phi_t$ is the quantum-evolved state which is obtained by the implementation of the proposed quantum time evolution operator ($\mathbf{U}(\Delta t)$) applied on the initial quantum state ($\psi_0$). The swap-test circuit is composed of Hadamard gates, control qubit, input states, and the swap gate as shown in Fig. \ref{swap_ckt}. We have performed our experiments for varying qubit size ($n$) to test the fidelity of the proposed quantum circuit and also compared our result with the Shokri et al. method \cite{shokri2021implementation}. For the fidelity experiments we have kept, the evolution time $\Delta t=0.1$ (with varying Trotter step size) and various measurement shots (depending on the qubit size) are performed on the IBM Qasm simulator.  It is observed that our proposed quantum circuit possesses a fidelity of $0.73$ with $3$ qubits, and it reaches a fidelity of $0.99$ with $9$ or more qubits. As compared to the state of the art ($0.89$ approximately with $9$ qubits), the proposed quantum circuit shows significant improvement in the fidelity (with a fidelity of $0.99$).


\begin{figure}[h!]
\centering
    \includegraphics[width=0.95\linewidth]{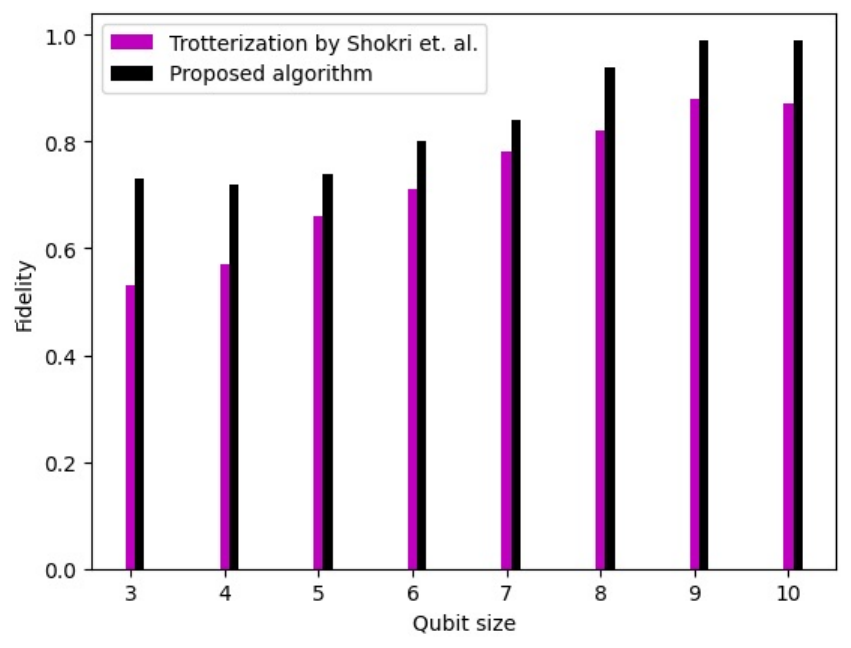}
\caption{Quantum fidelity with varying qubit size in IBM Qasm simulator}
\label{fid_res}
\end{figure}

\section{Computational complexity and error analysis}
The computational gate complexity of the Trotter-Suzuki method for the Hamiltonian simulation is of $\Theta(n^2)$ with a qubit size of $n$. In recent literature Shokri et. al. \cite{shokri2021implementation}, have shown an implementation (up to $5$ qubits) which can take a total number of $3n + ^n\mathcal{C}_2$ gates.  Our proposed QPA-based algorithm with the QATE encoding method further reduces the complexity, given in the below Lemma. 
\begin{lemma}\label{gate_compl}
    Given an $n$-qubit quantum circuit, the QATE algorithm requires $\mathcal{O}(n)$ $1$-qubit gate, and $^{n-1}\mathcal{C}_2 + 2(n-2)$ number of $2$-qubit gates to design a bi-symmetric diagonal evolution operator. 
\end{lemma}

\begin{proof}
    The number of $Cx$ gates required for the exchange operator as discussed in Proposition-\ref{prop1} to prepare the bi-symmetric pattern in the diagonal matrix is given by $2(n-1)$. The prime locations for an $n$-qubit quantum circuit are created with $n-1$ phase gates. The number of combinations of controlled phase gates that can be placed in the quantum circuit for generating the bi-symmetric pattern following the QATE algorithm is given by $^{n-1}\mathcal{C}_2$. Hence, the total number of single-qubit, and $2$-qubit quantum gates are given by $\mathcal{O}(n)$, and $^{n-1}\mathcal{C}_2 + 2(n-2)$. 
\end{proof}

The gate complexity as compared to \cite{shokri2021implementation} has been improved for $1$ qubit quantum gates. The overall fidelity of the quantum circuit is slightly increased in the proposed algorithm with the QATE encoding method. In problems, where one is interested in a specific region of the evolution operator, the QWE technique is preferred. To realize $n$-states, one can restrict the gate complexity to $\mathcal{O}(n)$ with the proper choice of phase and controlled phase gates. In fact, the least-square approach may also be adapted to find an approximation of the kinetic energy operator by revising our QWE approach with modified phase angles which can be a trade-off between approximation and complexity (between $\mathcal{O}(n)$ and $\mathcal{O}(n^2)$).   

\subsection{Gate counts}
The method by Shokri et. al. in \cite{shokri2021implementation} has shown the implementation of the time evolution operator on a real quantum machine for $4$ qubit register which takes a total number of $18$ quantum gates. We compare our proposed quantum algorithm with the QATE encoding method (which approximates the kinetic energy operator for all samples) with that of \cite{shokri2021implementation} (which is $3n+ ^n\mathcal{C}_2$ for $n$ qubit circuit). For a $4$ qubit quantum register circuit, our quantum circuit requires a total number of $12$ quantum gates, with a fidelity of $0.72$ approximately. With the QATE encoding method, we reduce the single qubit gate to $n$. The number of $2$-qubit gates is similar to the existing method. QWE encoding can take lesser $2$ qubit gates depending on the choice of window size. 

\subsection{Circuit depth}
We have performed circuit depth analysis which is an important measure of usage of gate resources. We have implemented our proposed algorithm and the existing method \cite{shokri2021implementation} on qiskit, and compare the circuit depth as shown in Table-\ref{ckt_depth}. The circuit depth is reduced in the proposed QATE encoding method as compared to the existing approach. 

\begin{table}[htb!]
    \centering
    \begin{tabular}{|c|c|c|}
    \hline
       \textbf{Qubit Size (n)}  & \textbf{Existing approach\cite{shokri2021implementation}} & \textbf{Proposed QATE algorithm}  \\ \hline
        3 & 16& 9\\ \hline 
         4  & 24& 18 \\ \hline
          5   & 32& 22 \\ \hline
           6    & 40& 36 \\ \hline
    \end{tabular}
    \caption{Circuit depth}
    \label{ckt_depth}
\end{table}

\subsection{Error analysis}
There are several sources of errors in the practical circuit simulation on a quantum machine, such as gate-level errors, cross talk, readout and coupling errors, and simulation errors. We have seen that the gate level error is significant for the $2$-qubit gates (example: CNOT). In the QATE encoding method, we have used $1$-qubit and $2$-qubit gates for approximating the kinetic energy function. In this approximation, we incur residual error of $\mathcal{O}\left(h^{3}\right)$ where $h$ is the distance between two successive samples (also called step size). For example, if $10$ qubits are taken to encode the potential energy within a region $[0,1]$, the $h$ can be approximately equal to $1/{2^{10}} \approx 9.7\times 10^{-4}$.

The overall approximate error bound in the diagonal unitary encoding of the function $f(x)$ using the proposed polynomial encoding procedure within polynomial order $r$ encoded in $n$ qubit registers and evolved for time $\Delta t$ has the form given as follows:
     \begin{align}
        \Vert \epsilon_s \Vert \approx \mathcal{O}\left(h^{3}\right) + L_2 \sigma_g^2 + \mathcal{O}\left(1+ \Delta t \left(\frac{T_1+T_2}{T_1 T_2}\right) \right) + \sigma_{cr}^2.
        \label{err_bound1}
    \end{align},
where $\mathcal{O}\left(h^{3}\right)$ is the residual error for step size $h$, $L_2$ represents total CNOT gates with each variance of $\sigma_g^2$, $T_1,T_2$ are decoherence time constants (here, we have taken first order approximation of the decoherence term), $\Delta t$ be evolution time, $\sigma_{cr}^2$ denotes total read-out error variance term.


\section{Conclusion}

In this research article, we have studied quantum time evolution operator design on a quantum machine considering practical constraints.  Time evolution plays a vital role in diverse disciplines for studying dynamics, especially in atomic chemistry. Considering the first quantization level, we have proposed a Hamiltonian encoding method for the Kinetic energy operator. It improves the total gate counts and fidelity for the time evolution of Gaussian wave packets.  Further, we have proposed a quantum architecture namely quantum pyramid architecture to efficiently simulate the kinetic energy taking half of the sampled values on a quantum computer by exploiting its structural aspects. The underlying mathematical propositions are given with examples in the appendix. While the proposed QATE encoding shows a time evolution process with high accuracy, the application of the proposed QWE is unknown at the moment. QWE method can be a future direction of research as it exploits the complexity benefit in the momentum domain. We show the complexity analysis of the proposed quantum algorithm with gate counts for $1$ qubit and $2$ qubit gates. Experimental results are shown on the IBM quantum simulator, and the fidelity is compared with the state of the art. There are several future directions of this research for the study of dynamics in chemical experiments, free particle systems, multi-body systems etc.

\section{Acknowledgement}

We acknowledge Rajiv Sangle, MTech in Quantum technology at Indin Institute of Science for his support in the Swap test circuit. We acknowledge Anupama Ray, Dhiraj Madan, and SheshaShayee K Raghunathan of IBM Research Bangalore for their valuable suggestions for improving our work.


\ifCLASSOPTIONcaptionsoff
  \newpage
\fi

\bibliographystyle{IEEEtran}
\bibliography{reference}%

\vspace{4mm}
\section{\centering Appendix-I}
\vspace{4mm}


\subsection{Proof of Lemma \ref{lemma1} }\label{bisym-lem1}
    Given $\mathbf{C}$ be a CNOT operator, and $\mathbf{P}=\mathbf{I}\otimes \mathbf{P}_1$ is another operator with $\mathbf{I}$ be the identity operator and $\mathbf{P}_1$ is some phase gate, then $\mathbf{R}=\mathbf{C} \mathbf{P} \mathbf{C}^{\dagger}$ is a bi-symmetric quantum operator. 

\begin{proof}
The CNOT gate $\mathbf{C}$ is denoted by
\begin{align}
    \mathbf{C}=\begin{bmatrix}
1 & 0 & 0 & 0 \\
0 & 1 & 0 & 0 \\
0 & 0 & 0 & 1 \\
0 & 0 & 1 & 0 \\
\end{bmatrix}
\end{align}

Let, the phase gate $\mathbf{P}_1$ is parameterised with the phase $\theta_1$. Hence, the matrix $\mathbf{P}$ is given by
\begin{align}
    \mathbf{P}&=\mathbf{I}\otimes \mathbf{P}_1\nonumber\\
    &=\begin{bmatrix}
1 & 0 \\
0 & 1 \\
\end{bmatrix}
\otimes
\begin{bmatrix}
1 & 0 \\
0 & e^{i\theta_1} \\
\end{bmatrix}\nonumber\\
&=\begin{bmatrix}
1 & 0 & 0 & 0 \\
0 & e^{i\theta} & 0 & 0 \\
0 & 0 & 1 & 0 \\
0 & 0 & 0 & e^{i\theta_1}\\
\end{bmatrix}
\end{align}

The matrix $\mathbf{R}$ can be written as,
\begin{align}
\mathbf{R} &= \mathbf{C} \mathbf{P} \mathbf{C}' \nonumber\\
&=
\begin{bmatrix}
1 & 0 & 0 & 0 \\
0 & 1 & 0 & 0 \\
0 & 0 & 0 & 1 \\
0 & 0 & 1 & 0 \\
\end{bmatrix}
\begin{bmatrix}
1 & 0 & 0 & 0 \\
0 & e^{i\theta} & 0 & 0 \\
0 & 0 & 1 & 0 \\
0 & 0 & 0 & e^{i\theta} \\
\end{bmatrix}
\begin{bmatrix}
1 & 0 & 0 & 0 \\
0 & 1 & 0 & 0 \\
0 & 0 & 0 & 1 \\
0 & 0 & 1 & 0 \\
\end{bmatrix}\nonumber\\
&=\begin{bmatrix}
1 & 0 & 0 & 0 \\
0 & 1 & 0 & 0 \\
0 & 0 & 0 & 1 \\
0 & 0 & 1 & 0 \\
\end{bmatrix} \begin{bmatrix}
1 & 0 & 0 & 0 \\
0 & e^{i\theta} & 0 & 0 \\
0 & 0 & 0 & 1 \\
0 & 0 & e^{i\theta} & 0 \\
\end{bmatrix}\nonumber\\
&=\begin{bmatrix}
1 & 0 & 0 & 0 \\
0 & e^{i\theta} & 0 & 0 \\
0 & 0 & e^{i\theta} & 0 \\
0 & 0 & 0 & 1 \\
\end{bmatrix}.
\end{align}
Here, $R$ is a bi-symmetric operator as it is symmetric about both of its main diagonals. 
\end{proof}

\subsection{Proof of Lemma-\ref{lemma2}}\label{bisym-lem2}
    Given a list of phase gates as $\mathbf{P}_1,~\dots,~\mathbf{P}_n$ with every $\mathbf{P}_j=diag([1 ~ e^{i\theta_j}])$ placed at $j^{th}$ qubit starting $q[1]$ (second qubit) to $q[n-1]$ (last qubit) with $n=\log_2 N$, the product of the operators $\mathbf{F}_1,~\dots, ~\mathbf{F}_n$ is a diagonal matrix of dimension $N\times N$ with first $2^{n-1}$ elements repeated in order along the main diagonal, where every $\mathbf{F}_j$ is obtained by placing $\mathbf{P}_j$ phase gate at $j^{th}$ qubit in absence of any other gates. 

\begin{proof}
    Let us take a $3$-qubit quantum system, and we place a phase gate $\mathbf{P}_1$, at the $2^{nd}$ qubit (i.e, $q[1]$ starting the count from $0$). The effective operator can be written as
    \begin{align}
        \mathbf{F}_1&= \mathbf{I}\otimes\mathbf{P}_1\otimes \mathbf{I}\nonumber\\
        &=\begin{bmatrix}
            1 & 0\\0&1
        \end{bmatrix}\otimes\begin{bmatrix}
            1 & 0\\0&e^{i\theta_1}
        \end{bmatrix}\otimes\begin{bmatrix}
            1 & 0\\0&1
        \end{bmatrix}\nonumber\\
    &=\begin{bmatrix}
            1 & 0\\0&1
        \end{bmatrix}\otimes\begin{bmatrix}
          1&0&0&0\\0&1&0&0\\0&0&e^{i\theta_1}&0\\0&0&0&e^{i\theta_1}
        \end{bmatrix}\nonumber\\
        &=\begin{bmatrix}
1&0&0&0&0&0&0&0\\0&1&0&0&0&0&0&0\\0&0&e^{i\theta_1}&0&0&0&0&0\\0&0&0&e^{i\theta_1}&0&0&0&0\\0&0&0&0&1&0&0&0\\0&0&0&0&0&1&0&0\\0&0&0&0&0&0&e^{i\theta_1}&0\\0&0&0&0&0&0&0&e^{i\theta_1}
        \end{bmatrix}
\end{align}
Similarly, placing another phase gate $\mathbf{P}_2$ in $3rd$ qubit yields the operator,
\begin{align}
     \mathbf{F}_2&= \mathbf{I}\otimes\mathbf{I}\otimes\mathbf{P}_2 \nonumber\\
     &=\begin{bmatrix}
            1 & 0\\0&1
        \end{bmatrix}\otimes\begin{bmatrix}
            1 & 0\\0&1
        \end{bmatrix}\otimes\begin{bmatrix}
            1 & 0\\0&e^{i\theta_2}
        \end{bmatrix}\nonumber\\
    &=\begin{bmatrix}
1&0&0&0&0&0&0&0\\0&e^{i\theta_2}&0&0&0&0&0&0\\0&0&1&0&0&0&0&0\\0&0&0&e^{i\theta_2}&0&0&0&0\\0&0&0&0&1&0&0&0\\0&0&0&0&0&e^{i\theta_2}&0&0\\0&0&0&0&0&0&1&0\\0&0&0&0&0&0&0&e^{i\theta_2}
        \end{bmatrix}    
\end{align}
Now, while both phase gates are placed together on the quantum circuit, the effective operator becomes (\ref{8x8eqn}).
\begin{figure}[htb!]
\scriptsize
\begin{align}
    \mathbf{F}_1\mathbf{F}_2 &=  \begin{bmatrix}
1&0&0&0&0&0&0&0\\0&e^{i\theta_2}&0&0&0&0&0&0\\0&0&e^{i\theta_1}&0&0&0&0&0\\0&0&0&e^{i(\theta_1+\theta_2)}&0&0&0&0\\0&0&0&0&1&0&0&0\\0&0&0&0&0&e^{i\theta_2}&0&0\\0&0&0&0&0&0&e^{i\theta_1}&0\\0&0&0&0&0&0&0&e^{i(\theta_1+\theta_2)}.
        \end{bmatrix}
        \label{8x8eqn}
\end{align}
\hrule
\end{figure}
In a similar way, if we increase the number of input qubits and place the phase gates from second qubit onward (keeping no gate on the first qubit, i.e., $q[0]$), we can create the product of the operators $\mathbf{F}_1,\mathbf{F}_2,\dots,\mathbf{F}_N$ to be an operator where along the diagonal first $2^{N-1}$ elements are repeated in the second half. For $n=3$, the first $4$ elements are repeated in the next half along the diagonal in exact order. 
\end{proof}

\subsection{Proof of Proposition-\ref{prop1}}\label{propsym1}

    The operator $\mathbf{A}=\mathbf{I}\otimes \dots \otimes (\ket{\mathbf{0}}\bra{\mathbf{0}} \times \mathbf{I}~ + ~\ket{\mathbf{1}}\bra{\mathbf{1}}\times \mathbf{X})\otimes\mathbf{I}\otimes\dots \otimes\mathbf{I}$ is a row-exchange operator for a given matrix $\mathbf{F}$ (assuming compatible with $\mathbf{A}$) and a Pauli-operator $\mathbf{X}$, if it is multiplied as $\mathbf{AF}$, and $\mathbf{A}^{\dagger}=\mathbf{A}$ is a column-exchange operator when it is post-multiplied as $\mathbf{FA^{\dagger}}$. The product $\mathbf{A}\mathbf{F}\mathbf{A}^{\dagger}$ has a symmetry about the mid-point along the main diagonal when $\mathbf{F}=\mathbf{F}_1,\dots, \mathbf{F}_n$ following lemma-\ref{lemma2}.

\begin{proof}

In the quantum pyramid architecture, the ladder of CNOT gates (in the left) placed on second qubit  connecting the first qubit has the form $\mathbf{A}=(\ket{\mathbf{0}}\bra{\mathbf{0}} \times \mathbf{I}~ + ~\ket{\mathbf{1}}\bra{\mathbf{1}}\times \mathbf{X})\otimes\mathbf{I}\dots \otimes \mathbf{I}$. Similarly, for a CNOT gate placed on $k^{th}$ qubit, the composite representation of the operator becomes $\mathbf{A}=\mathbf{I}\otimes \dots \mathbf{I}\otimes(\ket{\mathbf{0}}\bra{\mathbf{0}} \times \mathbf{I}~ + ~\ket{\mathbf{1}}\bra{\mathbf{1}}\times \mathbf{X})\otimes\mathbf{I}\otimes \dots \mathbf{I}$, where $(\ket{\mathbf{0}}\bra{\mathbf{0}} \times \mathbf{I}~ + ~\ket{\mathbf{1}}\bra{\mathbf{1}}\times \mathbf{X})$ is placed in the $k^{th}$ position in the operator prepared by by tensor product of CNOT and identity operators. The overall operator by the ladders of CNOT is the product of all such composite operators. As an example, it has the form  for $3$ qubits as follows 
\begin{align}
    \mathbf{A}= \begin{bmatrix}
        1&0&0&0&0&0&0&0\\
        0&1&0&0&0&0&0&0\\
        0&0&1&0&0&0&0&0\\
        0&0&0&1&0&0&0&0\\
        0&0&0&0&0&0&0&1\\
        0&0&0&0&0&0&1&0\\
        0&0&0&0&0&1&0&0\\
        0&0&0&0&1&0&0&0
    \end{bmatrix}.
    \label{row-exc}
\end{align}
Similarly, we get another operator $\mathbf{A}^{\dagger}=\mathbf{A}$ in the right side of the QPA with the ladders of CNOT gate. Note that, the multiplication of $\mathbf{A}$ on the left of $\mathbf{F}$, i.e., $\mathbf{AF}$ will exchange the rows of $\mathbf{F}$ (last $N/2$ rows here), and the multiplication of $\mathbf{A}$ on the right side of $\mathbf{F}$, i.e., $\mathbf{F}\mathbf{A}$ will exchange the columns of $\mathbf{F}$ (last $N/2$ columns here). Through successive row and column exchange in $\mathbf{F}$, we get a symmetry in the diagonal of the operator $\mathbf{F}$. As an example, for $3$ qubit system $\mathbf{AF}\mathbf{A}^{\dagger}$ ( assuming $\theta_1+\theta_2=\theta_3$) has the following form (\ref{afa_eqn}).
\begin{figure*}
\hrule
\begin{align}
    \mathbf{AF}\mathbf{A}^{\dagger} &= \begin{bmatrix}
        1&0&0&0&0&0&0&0\\
        0&1&0&0&0&0&0&0\\
        0&0&1&0&0&0&0&0\\
        0&0&0&1&0&0&0&0\\
        0&0&0&0&0&0&0&1\\
        0&0&0&0&0&0&1&0\\
        0&0&0&0&0&1&0&0\\
        0&0&0&0&1&0&0&0
    \end{bmatrix} \begin{bmatrix}
1&0&0&0&0&0&0&0\\0&e^{i\theta_2}&0&0&0&0&0&0\\0&0&e^{i\theta_1}&0&0&0&0&0\\0&0&0&e^{i\theta_3}&0&0&0&0\\0&0&0&0&1&0&0&0\\0&0&0&0&0&e^{i\theta_2}&0&0\\0&0&0&0&0&0&e^{i\theta_1}&0\\0&0&0&0&0&0&0&e^{i\theta_3}\end{bmatrix} \begin{bmatrix}
        1&0&0&0&0&0&0&0\\
        0&1&0&0&0&0&0&0\\
        0&0&1&0&0&0&0&0\\
        0&0&0&1&0&0&0&0\\
        0&0&0&0&0&0&0&1\\
        0&0&0&0&0&0&1&0\\
        0&0&0&0&0&1&0&0\\
        0&0&0&0&1&0&0&0
    \end{bmatrix} \nonumber\\
    &= \begin{bmatrix}
        1&0&0&0&0&0&0&0\\
        0&1&0&0&0&0&0&0\\
        0&0&1&0&0&0&0&0\\
        0&0&0&1&0&0&0&0\\
        0&0&0&0&0&0&0&1\\
        0&0&0&0&0&0&1&0\\
        0&0&0&0&0&1&0&0\\
        0&0&0&0&1&0&0&0
    \end{bmatrix} \begin{bmatrix}
1&0&0&0&0&0&0&0\\0&e^{i\theta_2}&0&0&0&0&0&0\\0&0&e^{i\theta_1}&0&0&0&0&0\\0&0&0&e^{i\theta_3}&0&0&0&0\\0&0&0&0&0&0&0&1\\0&0&0&0&0&0&e^{i\theta_2}&0\\0&0&0&0&0&e^{i\theta_1}&0&0\\0&0&0&0&e^{i\theta_3}&0&0&0\end{bmatrix} \nonumber\\
    &= \begin{bmatrix}
1&0&0&0&0&0&0&0\\0&e^{i\theta_2}&0&0&0&0&0&0\\0&0&e^{i\theta_1}&0&0&0&0&0\\0&0&0&e^{i\theta_3}&0&0&0&0\\0&0&0&0&e^{i\theta_3}&0&0&0\\0&0&0&0&0&e^{i\theta_1}&0&0\\0&0&0&0&0&0&e^{i\theta_2}&0\\0&0&0&0&0&0&0&1\end{bmatrix}.
\label{afa_eqn}
\end{align}
\hrule
\end{figure*}

Hence, $\mathbf{AF}\mathbf{A}^{\dagger}$ is a bi-symmetric operator. Now, one can encode the phases using a suitable algorithm as per the desired unitary required for the quantum time evolution.     
\end{proof}

\textbf{Note:} In QISKIT, the orientation of the LSB and MSB is different. As such the matrix that comes out is also different from the ones shown here. However, the end outcome is the same for the final quantum state.

\end{document}